\documentclass[11pt]{article}
\usepackage{amsmath,amssymb,amsfonts,amsthm,epsfig,bm,xspace,epsfig,latexsym,hyperref}
\newtheorem{theorem}{Theorem}[section]
\newtheorem{lemma}[theorem]{Lemma}

\newtheorem{proposition}[theorem]{Proposition}
\newtheorem{corollary}[theorem]{Corollary}

\newtheorem{definition}{Definition}
\theoremstyle{definition}

\newcommand{\ignore}[1]{[{\tiny *ignored part*}]}

\renewcommand{\Pr}{\mathop{\bf Pr\/}}
\newcommand{\E}{\mathop{\bf E\/}}

\newcommand{\R}{\mathbb R}
\newcommand{\N}{\mathbb N}

\newcommand{\eps}{\epsilon}

\newcommand{\dist}{\mathrm{dist}}

\newcommand{\la}{\langle}
\newcommand{\ra}{\rangle}

\newcommand{\wh}{\widehat}
\newcommand{\wt}{\widetilde}

\newcommand{\calH}{{\cal H}}

\newcommand{\bx}{\boldsymbol{x}}
\newcommand{\by}{\boldsymbol{y}}

\newcommand{\bh}{\boldsymbol{h}}

\newcommand{\Stab}{\mathbb{S}}
\newcommand{\K}{\mathbb{K}}

\begin{document}

\title{Optimal lower bounds for locality sensitive hashing \\ (except when \emph{q} is tiny)}

\author{ Ryan O'Donnell\thanks{Supported by NSF grants CCF-0747250 and CCF-0915893, BSF grant 2008477, and Sloan and Okawa fellowships.} \qquad Yi Wu \qquad Yuan Zhou\\Computer Science Department\\Carnegie Mellon University\\\{odonnell,yiwu,yuanzhou\}@cs.cmu.edu}

\date{}

\maketitle

\abstract
We study lower bounds for Locality Sensitive Hashing (LSH) in the strongest setting: point sets in $\{0,1\}^d$ under the Hamming distance.  Recall that $\calH$ is said to be an $(r, cr, p, q)$-sensitive hash family if all pairs $x,y \in \{0,1\}^d$ with $\dist(x,y) \leq r$ have probability at least $p$ of collision under a randomly chosen $h \in \calH$, whereas all pairs $x,y \in \{0,1\}^d$ with $\dist(x,y) \geq cr$ have probability at most $q$ of collision.  Typically, one considers $d \to \infty$, with $c > 1$ fixed and $q$ bounded away from $0$. 

For its applications to approximate nearest neighbor search in high dimensions, the quality of an LSH family $\calH$ is governed by how small its ``rho parameter'' $\rho = \ln(1/p)/\ln(1/q)$ is as a function of the parameter~$c$. The seminal paper of Indyk and Motwani showed that for each $c \geq 1$, the extremely simple family $\calH = \{x \mapsto x_i : i \in d\}$ achieves $\rho \leq 1/c$.  The only known lower bound, due to Motwani, Naor, and Panigrahy, is that $\rho$ must be at least $(e^{1/c}-1)/(e^{1/c}+1)\geq .46/c$ (minus $o_d(1)$).

In this paper we show an optimal lower bound: $\rho$ must be at least $1/c$ (minus $o_d(1)$).  This lower bound for Hamming space yields a lower bound of $1/c^2$ for Euclidean space (or the unit sphere) and $1/c$ for the Jaccard distance on sets; both of these match known upper bounds.  Our proof is simple; the essence is that the noise stability of a boolean function at $e^{-t}$ is a log-convex function of $t$.  

Like the Motwani--Naor--Panigrahy lower bound, our proof relies on the assumption that $q$ is not ``tiny'', meaning of the form $2^{-\Theta(d)}$. Some lower bound on $q$ is always necessary, as otherwise it is trivial to achieve $\rho = 0$. The range of $q$ for which our lower bound holds is the same as the range of $q$ for which $\rho$ accurately reflects an LSH family's quality.  Still, we conclude by discussing why it would be more satisfying to find LSH lower bounds that hold for tiny $q$.

\setcounter{page}{0} \thispagestyle{empty}
\newpage

\section{Locality Sensitive Hashing} \label{sec:lsh}
Locality Sensitive Hashing (LSH) is a widely-used algorithmic tool which brings the classic technique of hashing to geometric settings.  It was
introduced for general metric spaces in the seminal work of Indyk and Motwani~\cite{IM98}. Indyk and Motwani showed that the important problem of (approximate) nearest neighbor search can be reduced to the problem of devising good LSH families.  Subsequently, numerous papers demonstrating the practical utility of solving high-dimensional nearest neighbor search problems via the LSH approach~\cite{GIM99,Buh01,CDFGIMUY01,SVD03,RPH05,DDGR07}.  For a survey on LSH, see Andoni and Indyk~\cite{AI08}.\\

We recall the basic definition from~\cite{IM98}:
\begin{definition}
Let $(X, \dist)$ be a distance space\footnote{A metric space where the triangle inequality need not hold.}, and let $U$ be any finite or countably infinite set.  Let $r > 0$, $c > 1$.  A probability distribution $\calH$ over functions $h : X \to U$  is \emph{$(r, cr, p, q)$-sensitive} if for all $x, y \in X$,
\begin{eqnarray*}
\dist(x,y) \leq r &\Rightarrow& \Pr_{\bh \sim \calH}[\bh(x) = \bh(y)] \geq p,\\
\dist(x,y) \geq cr &\Rightarrow&  \Pr_{\bh \sim \calH}[\bh(x) = \bh(y)] \leq q,
\end{eqnarray*}
where $q < p$. We often refer to $\calH$ as a \emph{locally sensitive hash (LSH) family} for $(X,\dist)$.
\end{definition}

As mentioned, the most useful application of LSH is to the \emph{approximate near neighbor problem} in high dimensions:
\begin{definition} For a set of $n$ points $P$ in a metric space $(X, \dist)$, the \emph{$(r,c)$-near neighbor problem} is to process the points into a data structure that supports the following type of query: given a point $x \in X$, if there exists $y \in P$ with $\dist(x,y) \leq r$, the data structure should return a point $z \in P$ such that $\dist(x,z) \leq cr$.
\end{definition}
Several important problems in computational geometry reduce to the approximate near neighbor problem, including approximate versions of nearest neighbor, furthest neighbor, close pair, minimum spanning tree, and facility location.  For a short survey of these topics, see Indyk~\cite{Ind04}.\\

Regarding the reduction from $(r,c)$-near neighbor problem to LSH, it is usual (see~\cite{Ind01-thesis,DIIM04}) to credit roughly the following theorem to~\cite{IM98,GIM99}:
\begin{theorem} \label{thm:lsh-to-nn} Suppose $\calH$ is an $(r,cr,p,q)$-sensitive family for the metric space $(X,\dist)$.  Then one can solve the $(r,c)$-near neighbor problem with a (randomized) data structure that uses $O(n^{1+\rho} + dn)$ space and has query time dominated by $O(n^{\rho} \log_{1/q}(n))$ hash function evaluations.  (The preprocessing time is not much more than the space bound.)
\end{theorem}
Here we are using the following:
\begin{definition}  The \emph{rho parameter} of an $(r,cr,p,q)$-sensitive LSH family $\calH$ is
\[
\rho = \rho(\calH) = \frac{\ln(1/p)}{\ln(1/q)} \in (0,1).
\]
\end{definition}
Please note that in Theorem~\ref{thm:lsh-to-nn}, it is implicitly assumed~\cite{Ind09} that $q$ is bounded away from~$0$.  For ``subconstant'' values of $q$, the theorem does not hold. This point is discussed further in Section~\ref{sec:discussion}.\\

Because of Theorem~\ref{thm:lsh-to-nn}, there has been significant interest~\cite{DIIM04, TT07, AI08, Ney10} in determining the smallest possible $\rho$ that can be obtained for a given metric space and value of~$c$.  Constant factors are important here, especially for the most natural regime of $c$ close to~$1$.  For example, shrinking $\rho$ by an additive $.5$ leads to time and space savings of $\Theta(\sqrt{n})$.

\section{Previous work}
\subsection{Upper bounds}
The original work of Indyk and Motwani~\cite{IM98} contains the following simple yet strong result:
\begin{theorem} \label{thm:im} There is an LSH family $\calH$ for $\{0,1\}^d$ under the Hamming distance which for each $c > 1$ has rho parameter
\[
\rho(\calH) \leq \frac{1}{c},
\]
simultaneously for each $r < d/c$.
\end{theorem}
In this theorem, the family is simply the uniform distribution over the $d$ functions $h_i(x) = x_i$.  For a given $c$ and $r$, this family is obviously $(r, cr, 1 - r/d, 1 - cr/d)$-sensitive, whence
\[
\rho(\calH) = \frac{\ln(1/(1 - r/d))}{\ln(1/(1-cr/d))} \nearrow \frac{1}{c} \quad \text{as $r/d \to 0$}.
\]
We remark that the upper bound of $1/c$ in Theorem~\ref{thm:im} becomes tight only for asymptotically small $r/d$.  Indyk and Motwani showed that the same bound holds for the closely related ``Jaccard metric'' (see~\cite{IM98}), and also extended Theorem~\ref{thm:im} to an LSH family for the metric space $\ell_1$ (see also~\cite{AI06}). \\

Perhaps the most natural setting is when the metric space is the usual $d$-dimensional Euclidean space $\ell_2^d$.  Here, Andoni and Indyk~\cite{AI08} showed, roughly speaking, that $\rho \leq 1/c^2$:
\begin{theorem} \label{thm:ai} For any $r > 0$, $c > 1$, $d \geq 1$, there is a sequence of LSH families $\calH_t$ for $\ell_2^d$ satisfying
\[
\limsup_{t \to \infty} \rho(\calH_t) \leq \frac{1}{c^2}.
\]
(The complexity of evaluating a hash function $\bh \sim \calH_t$ also increases as $t$ increases.)
\end{theorem}

For other $\ell_s$ distance/metric spaces, Datar, Immorlica, Indyk, and Mirrokni~\cite{DIIM04} have similarly shown:\footnote{Please note that in~\cite{Pan06,MNP07} it is stated that~\cite{DIIM04} also improves the Indyk--Motwani $1/c$ upper bound for $\ell_1$ when $c \leq 10$.  However this is in error.}
\begin{theorem} \label{thm:diim} For any $r > 0$, $c > 1$, $d \geq 1$, and $0 < s < 2$, there is a sequence of LSH families $\calH_t$ for $\ell_s^d$ satisfying
\[
\limsup_{t \to \infty} \rho(\calH_t) \leq \max\left\{\frac{1}{c^s}, \frac{1}{c}\right\}.
\]
\end{theorem}

\noindent Other practical LSH families have been suggested for the Euclidean sphere~\cite{TT07} and $\ell_2$~\cite{Ney10}.

\subsection{Lower bounds}
There is one known result on lower bounds for LSH, due to Motwani, Naor, and Panigrahy~\cite{MNP07}:
\begin{theorem} \label{thm:mnp} Fix $c > 1$, $0 < q < 1$, and consider $d \to \infty$.  Then there exists some $r = r(d)$ such that for any LSH family $\calH$ for $\{0,1\}^d$ under Hamming distance which is $(r, cr, p, q)$-sensitive must satisfy
\[
\rho(\calH) \geq \frac{\exp(1/c)-1}{\exp(1/c) + 1} - o_{d}(1).
\]
\end{theorem}
The metric setting of $\{0,1\}^d$ under Hamming distance is the most powerful setting for lower bounds; as Motwani, Naor, and Panigrahy note, one can immediately deduce a lower bound of
\[
\frac{\exp(1/c^s)-1}{\exp(1/c^s) + 1} - o_{d}(1)
\]
for the setting of $\ell_s^d$.  This is simply because $\|x - y\|_s = \|x - y\|_1^{1/s}$ when $x, y \in \{0,1\}^d$.\\

As $c \to \infty$, the lower bound in Theorem~\ref{thm:mnp} approaches $\frac{1}{2c}$.  This is a factor of $2$ away from the upper bound of Indyk and Motwani.  The gap is slightly larger in the more natural regime of $c$ close to $1$; here one only has that $\rho(\calH) \geq \frac{e-1}{e+1} \frac{1}{c} \approx \frac{.46}{c}$.\\

Note that in Theorem~\ref{thm:mnp}, the parameter $q$ is fixed \emph{before} one lets $d$ tend to $\infty$; i.e., $q$ is assumed to be at least a ``constant''.  Even though this is the same assumption implicitly made in the application of LSH to near-neighbors (Theorem~\ref{thm:lsh-to-nn}), we feel it is not completely satisfactory.  In fact, as stated in~\cite{MNP07}, Theorem~\ref{thm:mnp} still holds so long as $q \geq 2^{-o(d)}$.  Our new lower bound for LSH also holds for this range of $q$.  But we believe the most satisfactory lower bound would hold even for ``tiny'' $q$, meaning $q = 2^{-\Theta(d)}$.  This point is discussed further in Section~\ref{sec:discussion}.\\

We close by mentioning the recent work of Panigrahy, Talwar, and Wieder~\cite{PTW08} which obtains a time/space lower bound for the $(r,c)$-near neighbor problem itself in several metric space settings, including $\{0,1\}^d$ under Hamming distance, and $\ell_2$.

\section{Our result}
In this work, we improve on Theorem~\ref{thm:mnp} by obtaining a sharp lower bound of $\frac{1}{c} - o_d(1)$ for every $c > 1$.  This dependence on $c$ is optimal, by the upper bound of Indyk and Motwani. The precise statement of our result is as follows:
\begin{theorem} \label{thm:main}  Fix $d \in \N$, $1 < c < \infty$, and $0 < q < 1$.  Then for a certain choice of $0 < \tau < 1$, any $(\tau d, c \tau d, p, q)$-sensitive hash family $\calH$ for $\{0,1\}^d$ under Hamming distance must satisfy
\begin{eqnarray}
\rho(\calH) \geq \frac{1}{c} - \wt{O}\left(\frac{\ln(2/q)}{d}\right)^{1/3}. \label{eqn:main}
\end{eqnarray}
Here, the precise meaning of the $\wt{O}(\cdot)$ expression is
\[
K \cdot \frac{\ln(2/q)}{d} \cdot \ln\left(\frac{d}{\ln(2/q)}\right),
\]
where $K$ is a universal constant, and we assume $d/\ln(2/q) \geq 2$, say.
\end{theorem}
As mentioned, the lower bound is only of the form $\frac{1}{c} - o_d(1)$ under the assumption that $q \geq 2^{-o(d)}$.  For $q$ of the form $2^{-d/B}$ for a large constant $B$, the bound~\eqref{eqn:main} still gives some useful information.\\

As with the Motwani--Naor--Panigrahy result, because our lower bound is for $\{0,1\}^d$  we may immediately conclude:
\begin{corollary} Theorem~\ref{thm:main} also holds for LSH families for the distance space $\ell_s$, $0 < s < \infty$, with the lower bound $1/c^s$ replacing $1/c$.
\end{corollary}
This lower bound matches the known upper bounds for Euclidean space $s = 2$ (\cite{AI08}) and $0 < s \leq 1$ (\cite{DIIM04}).  It seems reasonable to conjecture that it is also tight at least for $1 < s < 2$.\\

Finally, the lower bound in Theorem~\ref{thm:main} also holds for the Jaccard distance on sets, matching the upper bound of Indyk and Motwani~\cite{IM98}.  We explain why this is true in Section~\ref{sec:proof-idea}, although we omit the very minor necessary changes to the proof details.

\subsection{Noise stability}
Our proof of Theorem~\ref{thm:main} requires some  facts about boolean \emph{noise stability}.  We begin by recalling some basics of the analysis of boolean functions.
\begin{definition}  For $0 < \rho \leq 1$, we say that $(\bx,\by)$ are \emph{$\rho$-correlated random strings} in $\{0,1\}^d$ if $\bx$ is chosen uniformly at random and $\by$ is formed by rerandomizing each coordinate of $\bx$ independently with probability $1-\rho$.
\end{definition}
\begin{definition} Given $f : \{0,1\}^d \to \R$, the \emph{noise stability of~$f$ at $\rho$} is defined to be
\[
\Stab_f(\rho) = \E_{{\substack{(\bx, \by) \\ \text{$\rho$-correlated}}}}[f(\bx)f(\by)].
\]
We can extend the definition to functions $f : \{0,1\}^d \to \R^U$ via
\[
\Stab_f(\rho) = \E_{{\substack{(\bx, \by) \\ \text{$\rho$-correlated}}}}[\langle f(\bx), f(\by) \rangle],
\]
where $\la w, z \ra = \sum_{i \in U} w_i z_i$ is the usual inner product.\footnote{In the case that $U$ is countably infinite, we require our functions $f$ to have $\|f(x)\|_2 < \infty$ for all $x \in \{0,1\}^d$.}
\end{definition}
\begin{proposition}  \label{prop:ns-formula} Let $f : \{0,1\}^d \to \R^U$ and write $\wh{f}(S)$ for the usual Fourier coefficient of $f$ associated with $S \subseteq [d]$; i.e.,
\[
\wh{f}(S) = \frac{1}{2^d} \sum_{x \in \{0,1\}^d} f(x) \prod_{i \in S} (-1)^{x_i}  \in \R^U.
\]
Then
\[
\Stab_f(\rho) = \sum_{S \subseteq [d]} \|\wh{f}(S)\|_2^2 \rho^{|S|}.
\]
\end{proposition}
(This formula is standard when $f$ has range $\R$; see, e.g.,~\cite{O'D03}.  The case when $f$ has range $\R^U$ follows by repeating the standard proof.)\\

We are particularly interested in hash functions $h : \{0,1\}^d \to U$; we view these also as functions $\{0,1\}^d \to \R^U$ by identifying $i \in U$ with the vector $e_i \in \R^{U}$, which has a $1$ in the $i$th coordinate and a $0$ in all other coordinates.  Under this identification, $\la h(x),h(y) \ra$ becomes the $0$-$1$ indicator of the event $h(x) = h(y)$.  Hence for a fixed hash function $h$,
\begin{equation} \label{eqn:h-stab}
\Stab_h(\rho) = \Pr_{{\substack{(\bx, \by) \\ \text{$\rho$-correlated}}}}[h(\bx) = h(\by)].
\end{equation}
We also extend the notion of noise stability to hash \emph{families}:
\begin{definition} If $\calH$ is a hash family on $\{0,1\}^d$, we define
\[
\Stab_{\calH}(\rho) = \E_{\bh \sim \calH}[\Stab_{\bh}(\rho)].
\]
\end{definition}
By combining this definition with equation~\eqref{eqn:h-stab} and Proposition~\ref{prop:ns-formula}, we immediately deduce:
\begin{proposition}  \label{prop:main} Let $\calH$ be a hash family on $\{0,1\}^d$.  Then
\[
\Stab_{\calH}(\rho)\quad= \Pr_{\substack{\bh \sim \calH, \\(\bx, \by)\text{ $\rho$-corr'd}}}[\bh(\bx) = \bh(\by)] \quad=\quad \sum_{S \subseteq [d]} \E_{\bh \sim \calH}[\|\wh{\bh}(S)\|_2^2] \rho^{|S|}.
\]
\end{proposition}

Finally, it is sometimes more natural to express the parameter $\rho$ as $\rho = e^{-t}$, where $t \in [0,\infty)$.  (For example, we can think of a $\rho$-correlated pair $(\bx,\by)$ by taking $\bx$ to be uniformly random and $\by$ to be the string that results from running the standard continuous-time Markov Chain on~$\{0,1\}^d$, starting from $\bx$, for time $td$.)  We make the following definition:
\begin{definition} For $t \in [0,\infty)$, we define $\K_h(t) = \Stab_h(e^{-t})$, and we similarly define $\K_{\calH}(t)$.
\end{definition}

\subsection{The proof, modulo some tedious calculations} \label{sec:proof-idea}
We now present the essence of our proof of Theorem~\ref{thm:main}.  It will be quite simple to see how it gives a lower bound of the form $\frac{1}{c} - o_d(1)$ (assuming $q$ is not tiny). Some very tedious calculations (Chernoff bounds, elementary inequalities, etc.) are needed to get the precise statement given in Theorem~\ref{thm:main}; the formal proof is therefore deferred to Section~\ref{sec:proof}.\\

Let $\calH$ be a hash family on $\{0,1\}^d$, and let us consider
\begin{equation} \label{eqn:consider}
\K_{\calH}(t) = \Pr_{\substack{\bh \sim \calH, \\(\bx, \by)\text{ $e^{-t}$-corr'd}}}[\bh(\bx) = \bh(\by)].
\end{equation}
Let us suppose that $t$ is very small, in which case $e^{-t} \approx 1 - t$.  When $(\bx, \by)$ are $(1-t)$-correlated strings, it means that $\by$ is formed from the random string $\bx$ by rerandomizing each coordinate with probability $t$.  This is the same as flipping each coordinate with probability $t/2$.  Thus if we think of $d$ as large, a simple Chernoff bound shows that the Hamming distance $\dist(\bx,\by)$ will be very close to $(t/2)d$ with overwhelming probability.\footnote{Similarly, if we think of $\bx$ and $\by$ as subsets of $[d]$, their Jaccard distance will be very close to $t/(1+t/2) \approx t$ with overwhelming probability. With this observation, one obtains our lower bound on LSH families for the Jaccard distance on sets.} \\

Suppose now that $\calH$ is $((t/2)d + o(d), (ct/2)d -o(d), p, q)$-sensitive, so the distance ratio is $c - o_d(1)$.  In~\eqref{eqn:consider}, regardless of $\bh$ we will almost surely have $\dist(\bx,\by) \leq (t/2) + o(d)$; hence $\K_{\calH}(t) \geq p - o_d(1)$.  Similarly, we deduce $K_{\calH}(ct) \leq q + o_d(1)$.  Hence, neglecting the $o_d(1)$ terms, we get
\[
\rho(\calH) = \frac{\ln(1/p)}{\ln(1/q)} \gtrsim \frac{\ln(1/\K_{\calH}(t))}{\ln(1/\K_{\calH}(ct))}.
\]
We then deduce the desired lower bound of $1/c$ from the following theorem and its corollary:
\begin{theorem}  For any hash family $\calH$ on $\{0,1\}^d$, the function $\K_\calH(t)$ is log-convex in $t$.
\end{theorem}
\begin{proof} From Proposition~\ref{prop:main} we have
\[
\K_{\calH}(t) = \sum_{S \subseteq [d]} \E_{\bh \sim \calH}[\|\wh{\bh}(S)\|_2^2]e^{-t|S|}.
\]
Thus $\K_{\calH}(t)$ is log-convex, being a nonnegative linear combination of log-convex functions $e^{-t|S|}$.
\end{proof}
\begin{corollary} \label{cor:log-conc} For any hash family $\calH$ on $\{0,1\}^d$, $t \geq 0$, and $c \geq 1$,
\[
\frac{\ln(1/\K_{\calH}(t))}{\ln(1/\K_{\calH}(ct))} \geq \frac{1}{c}.
\]
\end{corollary}
\begin{proof}
By log-convexity, $\K_{\calH}(t)  \leq \K_{\calH}(ct)^{1/c} \cdot K_{\calH}(0)^{1-1/c} = \K_{\calH}(ct)^{1/c}$.  Here we used the fact that $K_{\calH}(0)=1$, which is immediate from the definitions because $e^{-0}$-correlated strings are always identical. The result follows.
\end{proof}

As mentioned, we give the careful proof keeping track of approximations in Section~\ref{sec:proof}.  But first, we note what we view as a shortcoming of the proof: after deducing $K_{\calH}(ct) \geq q - o_d(1)$, we wish to ``neglect'' the additive $o_d(1)$ term.  This requires that $o_d(1)$ indeed be negligible compared to $q$!  Being more careful, the $o_d(1)$ arises from a Chernoff bound applied to a Binomial$(d, ct)$ random variable, where $t > 0$ is very small.  So to be more precise, the error term is of the form $\exp(-\eps d)$, and hence is only negligible if $q \geq 2^{-o(d)}$.

\section{Discussion} \label{sec:discussion}

\subsection{On the reduction from LSH to near neighbor data structures}
As described in Section~\ref{sec:lsh}, it is normally stated that the quality of an $(r,cr,p,q)$-sensitive LSH family $\calH$ is governed by $\rho = \ln(1/p)/\ln(1/q)$, and more specifically that $\calH$ can be used to solve the $(r,c)$-near neighbor problem with roughly $O(n^{1+\rho})$ space and query time $O(n^{\rho})$.  However, this involves the implicit assumption that $q$ is bounded away from~$0$.\\

It is easy to see that \emph{some} lower bound on $q$ is essential. Indeed, for any (finite, say) distance space $(X,\dist)$ there is a trivially ``optimal'' LSH family for any $r$ and $c$:  For each pair $x,y \in X$ with $\dist(x,y) \leq r$, define $h_{x,y}$ by setting $h_{x,y}(x) = h_{x,y}(y) = 0$ and letting $h_{x,y}(z)$ have distinct positive values for all $z \neq x,y$.  If $\calH$ is the uniform distribution over all such $h_{x,y}$, then $p > 0$ and $q = 0$, leading to $\rho(\calH) = 0$.\\

To see why this trivial solution is not useful, and what lower bound on $q$ is desirable, we recall some aspects of the Indyk--Motwani reduction from LSH families to $(r,c)$-near neighbor data structures.  Suppose one wishes to build an $(r,c)$-near neighbor data structure for an $n$-point subset $P$ of the metric space $(X, \dist)$.  The first step in~\cite{IM98} is to apply the following:

\paragraph{Powering Construction:} Given an $(r, cr, p, q)$-sensitive family $\calH$ of functions $X \to U$ and a positive integer $k$, we define the family $\calH^{\otimes k}$ by drawing $\bh_1, \dots, \bh_k$ independently from $\calH$ and forming the function $\bh : X \to U^k$, $\bh(x) = (\bh_1(x), \dots, \bh_k(x))$.  It is easy to check that $\calH^{\otimes k}$ is $(r, cr, p^k, q^k)$-sensitive.\\

Indyk and Motwani show that if one has an $(r, cr, p', q')$-sensitive hash family with $q' \leq 1/n$, then one can obtain a $(r,c)$-near neighbor data structure with space roughly $O(n/p')$ and query time roughly $O(1/p')$.  Thus given an arbitrary $(r,cr,p,q)$-sensitive family $\calH$, Indyk and Motwani suggest using the Powering Construction with $k = \log_{1/q}(n)$.  The resulting  $\calH^{\otimes k}$ is $(r, cr, p', 1/n)$-sensitive, with $p' = p^k = n^{-\rho}$, yielding an $O(n^{1+\rho})$ space, $O(n^{\rho})$ time data structure.\\

However this argument makes sense only if $k$ is a positive integer.  For example, with the trivially ``optimal'' LSH family, we have $q = 0$ and thus $k = -\infty$.  Indeed, whenever $q \leq 1/n$ to begin with, one doesn't get $O(n^{1+\rho})$ space and $O(n^{\rho})$ time, one simply gets $O(n/p)$ space and $O(1/p)$ time.  For example, a hypothetical LSH family with $p = 1/n^{.5}$ and $q = 1/n^{1.5}$ has $\rho = 1/3$ but only yields an $O(n^{1.5})$ space, $O(n^{.5})$ time near neighbor data structure.\\

The assumption $q > 1/n$ is still not enough for the deduction in Theorem~\ref{thm:lsh-to-nn} to hold precisely.  The reason is that the Indyk--Motwani choice of $k$ may not be an integer.  For example, suppose we design an $(r, cr, p, q)$-sensitive family $\calH$ with $p = 1/n^{.15}$ and $q = 1/n^{.3}$.  Then $\rho = .5$.  However, we cannot actually get an $O(n^{1.5})$ space, $O(n^{.5})$ time data structure from this $\calH$.  The reason is that to get $q^k \leq 1/n$, we need to take $k = 4$.  Then $p^k = 1/n^{.6}$, so we only get an $O(n^{1.6})$ space, $O(n^{.6})$ time data structure.\\

The effect of rounding $k$ up to the nearest integer is not completely eliminated unless one makes the assumption, implicit in  Theorem~\ref{thm:lsh-to-nn}, that $q \geq \Omega(1)$.  Under the weaker assumption that $q \geq n^{-o(1)}$, the conclusion of Theorem~\ref{thm:lsh-to-nn} remains true up to $n^{o(1)}$ factors.  To be completely precise, one should assume $q \geq 1/n$ and take $k = \lceil \log_{1/q}(n) \rceil$.  If we then use $k \leq \log_{1/q}(n) + 1$, the Powering Construction will yield an LSH family with $q' \leq 1/n$ and $p' = (n/q)^{-\rho}$.  In this way, one obtains a refinement of Theorem~\ref{thm:lsh-to-nn} with no additional assumptions:
\begin{theorem} \label{thm:lsh-to-nn2} Suppose $\calH$ is an $(r,cr,p,q)$-sensitive family for the metric space $(X,\dist)$.  Then for $n$-point subsets of $X$ (and assuming $q \geq 1/n$), one can solve the $(r,c)$-near neighbor problem with a (randomized) data structure that uses $n \cdot O((n/q)^{\rho}+ d)$ space and has query time dominated by $O((n/q)^{\rho} \log_{1/q}(n))$ hash function evaluations.
\end{theorem}

\subsection{On assuming $q$ is not tiny}
Let us return from the near-neighbor problem to the study of locality sensitive hashing itself.  Because of the ``trivial'' LSH family, it is essential to impose some kind of lower bound on how small the parameter $q$ is allowed to be.  Motwani, Naor, and Panigrahy carry out their lower bound for LSH families on $\{0,1\}^d$ under the assumption that $q \geq \Omega(1)$, but also note that it goes through assuming $q \geq 2^{-o(d)}$.  Our main result, Theorem~\ref{thm:main}, is also best when $q \geq 2^{-o(d)}$, and is only nontrivial assuming $q \geq 2^{-d/B}$ for a sufficiently large constant $B$.\\

One may ask what the ``correct'' lower bound assumed on $q$ should be.  For the Indyk--Motwani application to $(r,c)$-near neighbor data structures, the answer seems obvious: ``$1/n$''.  Indeed, since the Indyk--Motwani reduction immediately uses Powering to reduce the $q$ parameter down to $1/n$, the most meaningful LSH lower bounds would simply involve fixing $q = 1/n$ and trying to lower bound $p$.\\

There is an obvious catch here, though, which is that in the definition of LSH, there \emph{is no notion of ``$n$''}!  Still, in settings such as $\{0,1\}^d$ which have a notion of dimension, $d$, it seems reasonable to think that applications will have $n = 2^{\Theta(d)}$.  In this case, to maintain the Indyk--Motwani Theorem~\ref{thm:lsh-to-nn2} up to $n^{o(1)}$ factors one would require $q \geq 2^{-o(d)}$.  This is precisely the assumption that this paper and the Motwani--Naor--Panigrahy paper have made.  Still, we believe that the most compelling kind of LSH lower bound for $\{0,1\}^d$ would be nontrivial even for $q = 2^{-d/b}$ with a ``medium'' constant $b$, say $b = 10$.  We currently do not have such a lower bound.

\section{Proof details} \label{sec:proof}
We require the following lemma, whose proof follows easily from Proposition~\ref{prop:main} and the definition of hash family sensitivity:
\begin{lemma} \label{lem:rand} Let $\calH$ be an $(r, cr, p, q)$-sensitive hash family on $\{0,1\}^d$ and suppose $(\bx, \by)$ is a pair of $e^{-u}$-correlated random strings.  Then
\[
p(1 - \Pr[\dist(\bx, \by) > r]) \leq \K_{\calH}(u) \leq q + \Pr[\dist(\bx, \by) < cr].
\]
\end{lemma}

We now prove Theorem~\ref{thm:main}, which for convenience we slightly rephrase as follows:
\begin{theorem} Fix $d \in \N$, $1 < c < \infty$, and $0 < q < 1$.  Then for a certain choice of $0 < \eps < 1$, any $((\eps/c)d, \eps d, p, q)$-sensitive hash family for $\{0,1\}^d$ under Hamming distance must satisfy
\[
\rho = \frac{\ln(1/p)}{\ln(1/q)} \geq \frac{1}{c} - K \cdot \lambda(d,q)^{1/3},
\]
where $K$ is a universal constant,
\[
\lambda(d,q) = \frac{\ln(2/q)}{d}  \ln\left(\frac{d}{\ln(2/q)}\right),
\]
and we assume $d/\ln(2/q) \geq 2$, say.
\end{theorem}
\begin{proof}
Let $0 < \Delta = \Delta(c, d, q) < .005$ be a small quantity to be chosen later, and let $\eps = .005 \Delta$.  Suppose that $\calH$ is an $((\eps/c)d, \eps d, p, q)$-sensitive hash family for $\{0,1\}^d$.  Our goal is to lower bound $\rho = \ln(1/p)/\ln(1/q)$.  By the Powering Construction we may assume that $q \leq 1/e$, and hence will use  $\ln(1/q) \geq 1$ without further comment.  Define also $t = 2\eps(1+\Delta/2)$ and $c' = c(1+\Delta)$.

Let $(\bx_1, \by_1)$ be $\exp(-t/c')$-correlated random strings and let $(\bx_2, \by_2)$ be $\exp(-t)$-correlated random strings.  Using the two bounds in Lemma~\ref{lem:rand} separately, we have
\[
\K_{\calH}(t/c') \geq p(1 - e_1), \qquad
\K_{\calH}(t)  \leq  q + e_2,
\]
where
\[
e_1 = \Pr[\dist(\bx_1, \by_1) > (\eps/c) d], \qquad e_2 = \Pr[\dist(\bx_2, \by_2) < \eps d].
\]
By Corollary~\ref{cor:log-conc}, we have
\begin{equation} \label{eqn:main1}
\frac{1}{c'} \leq \frac{\ln\Bigl(1/\K_{\calH}(t/c')\Bigr)}{\ln\Bigl(1/\K_{\calH}(t)\Bigr)} \leq \frac{\ln\left(\frac{1}{p(1-e_1)}\right)}{\ln\left(\frac{1}{q+e_2}\right)} = \frac{\ln(1/p) + \ln(1/(1-e_1))}{\ln(1/q) + \ln(1/(1+ e_2/q))}.
\end{equation}
We will use the following estimates:
\begin{gather}
\frac{1}{c'} = \frac{1}{c(1+\Delta)} \geq \frac{1}{c}(1-\Delta) = \frac{1}{c} - \frac{\Delta}{c}, \label{eqn:est-c}\\
\ln(1/(1-e_1)) \leq 1.01 e_1, \label{eqn:est-p} \\
\ln(1/q) + \ln(1/(1+e_2/q)) \geq \ln(1/q) - e_2/q = \ln(1/q)\Bigl(1-\frac{e_2}{q\ln(1/q)}\Bigr). \label{eqn:est-q}
\end{gather}
For~\eqref{eqn:est-p} we made the following
\begin{equation} \label{eqn:assn1}
\textbf{assumption:} \qquad e_1 \leq .01.
\end{equation}
We will also ensure that the quantity in~\eqref{eqn:est-q} is positive by making the following
\begin{equation} \label{eqn:assn2}
\textbf{assumption:} \qquad e_2 < q \ln(1/q).
\end{equation}
Substituting the three estimates~\eqref{eqn:est-c}--\eqref{eqn:est-q} into~\eqref{eqn:main1} we obtain
\begin{eqnarray*}
\frac{1}{c} - \frac{\Delta}{c} \leq \frac{\ln(1/p) + 1.01e_1}{\ln(1/q)\Bigl(1-\frac{e_2}{q\ln(1/q)}\Bigr)} &\Rightarrow& \frac{\ln(1/p) + 1.01e_1}{\ln(1/q)} \geq \left(\frac{1}{c} - \frac{\Delta}{c}\right)\left(1-\frac{e_2}{q\ln(1/q)}\right) \\
&\Rightarrow& \frac{\ln(1/p)}{\ln(1/q)} \geq \frac{1}{c} - \frac{\Delta}{c}-\frac{e_2}{q\ln(1/q)} -  \frac{1.01e_1}{\ln(1/q)}.
\end{eqnarray*}
Thus we have established
\begin{equation} \label{eqn:main-bound}
\rho \geq \frac{1}{c} - e, \qquad \text{where } e = \frac{\Delta}{c} + \frac{1.01e_1}{\ln(1/q)} + \frac{e_2}{q\ln(1/q)}.
\end{equation}

We now estimate $e_1$ and $e_2$ in terms of $\Delta$ (and $\eps$), after which we will choose $\Delta$ so as to minimize $e$.  By definition, $e_1$ is the probability that a Binomial$(d, \eta_1)$ random variable exceeds $(\eps/c)d$, where $\eta_1 = (1-\exp(t/c'))/2$.  Let us select $\delta_1$ so that $(1+\delta_1)\eta_1  = \eps/c$. Thus
\[
\delta_1 = \frac{\eps}{c \eta_1} - 1 = \frac{2\eps/c}{1 - \exp(-t/c')} - 1 \geq \frac{2\eps/c}{t/c'} - 1 = \frac{1+\Delta}{1+\Delta/2} - 1 \geq .498\Delta.
\]
Here we used the definitions of $t$ and $c'$, and then the assumption $\Delta < .005$.  Using a standard Chernoff bound, we conclude
\begin{equation} \label{eqn:e1}
e_1 = \Pr[\text{Binomial}(d,\eta_1) > (1+\delta_1) \eta_1 d] < \exp\left(-\frac{\delta_1^2}{2+\delta_1} \eta_1 d\right) < \exp\left(-\frac{\Delta^2}{8.08} \eta_1 d\right),
\end{equation}
using the fact that $\delta^2/(2+\delta)$ is increasing in $\delta$, and $\Delta < .005$ again. We additionally estimate
\[
\eta_1 = \frac{1-\exp(t/c')}{2} \geq \frac{t/c' - (t/c')^2/2}{2} = (t/2c') - (t/2c')^2 \geq .99(t/2c') = .99 \frac{\eps}{c}\left(\frac{1+\Delta/2}{1+\Delta}\right) \geq .98\frac{\eps}{c}.
\]
Here the second inequality used $t/2c' \leq .01$, which certainly holds since $t/2c' \leq \eps = .005 \Delta$. The third inequality used $\Delta \leq .005$.  Substituting this into~\eqref{eqn:e1} we obtain our upper bound for $e_1$,
\begin{equation} \label{eqn:e1f}
e_1 < \exp\left(-\frac{\Delta^2}{8.25} \frac{\eps}{c} d\right) = \exp\left(-\frac{.005\Delta^3}{8.25 c} d\right) < \exp\left(-\frac{\Delta^3}{2000 c} d\right).
\end{equation}

Our estimation of $e_2$ is quite similar:
\begin{equation} \label{eqn:e2}
e_2 = \Pr[\text{Binomial}(d,\eta_2) < (1-\delta_2) \eta_2 d] < \exp\left(-\frac{\delta_2^2}{2} \eta_2 d\right),
\end{equation}
where $\eta_2 = (1-\exp(-t))/2$ and $\delta_2$ is chosen so that $(1-\delta_2) \eta_2 = \eps$.  This entails
\[
\delta_2 = 1 - \frac{\eps}{\eta_2} = 1 - \frac{2\eps}{1-\exp(-t)} \geq 1 - \frac{2\eps}{t - t^2/2} = 1-\frac{1}{(t/2\eps) - \eps(t/2\eps)^2} = 1-\frac{1}{(1+\Delta/2) - \eps(1+\Delta/2)^2}.
\]
This expression is the reason we were forced to take $\eps$ noticeably smaller than $\Delta$.  Using our specific setting $\eps = .005\Delta$, we conclude
\[
\delta_2 \geq 1-\frac{1}{(1+\Delta/2) - \eps(1+\Delta/2)^2} = 1 - \frac{1}{1 + .495 \Delta -.005 \Delta^2 - .00125 \Delta^3} \geq .49\Delta,
\]
where we used $\Delta \leq .005$ again.  As for $\eta_2$, we can lower bound it similarly to $\eta_1$, obtaining
\[
\eta_2 \geq .99(t/2) = .99 \eps(1 + \Delta/2) \geq .99 \eps.
\]
Substituting our lower bounds for $\delta_2$ and $\eta_2$ into~\eqref{eqn:e2} yields
\begin{equation} \label{eqn:e2f}
e_2 < \exp\left(-\frac{(.49\Delta)^2}{2} \cdot .99\eps d\right) < \exp\left(-\frac{\Delta^3}{2000} d\right).
\end{equation}

Plugging our upper bounds~\eqref{eqn:e1f},~\eqref{eqn:e2f} for $e_1$, $e_2$ into~\eqref{eqn:main-bound} gives
\begin{equation} \label{eqn:last}
e = \frac{\Delta}{c} + \frac{1.01\exp(-\frac{\Delta^3}{2000c} d)}{\ln(1/q)} + \frac{\exp(-\frac{\Delta^3}{2000} d)}{q\ln(1/q)}.
\end{equation}
Finally, we would like to choose
\[
\Delta = K_1 c^{1/3} \lambda(d,q)^{1/3},
\]
where $K_1$ is an absolute constant.  For $K_1$ sufficiently large, this makes all three terms in the bound~\eqref{eqn:last} at most
\[
2K_1 \lambda(d,q)^{1/3} =  \wt{O}\left(\frac{\ln(2/q)}{d}\right)^{1/3}.
\]
This would establish the theorem.

It only remains to check whether this is a valid choice for $\Delta$.  First, we note that with this choice, assumptions~\eqref{eqn:assn1} and~\eqref{eqn:assn2} follow from~\eqref{eqn:e1f} and~\eqref{eqn:e2f} (and increasing $K_1$ if necessary).  Second, we required that $\Delta \leq .005$.  This may not hold.  However, if it fails then we have
\[
\lambda(d,q)^{1/3} > \frac{.005}{K_1 c^{1/3}}.
\]
We can then trivialize the theorem by taking $K = (K_1/.005)^3$,  making the claimed lower bound for $\rho$ smaller than $1/c - 1/c^{1/3} \leq 0$.
\end{proof}

\section*{Acknowledgments}
The authors would like to thank Alexandr Andoni, Piotr Indyk, Assaf Naor, and Kunal Talwar for helpful discussions.

\bibliographystyle{alpha}

\bibliography{lsh}

\newcommand{\etalchar}[1]{$^{#1}$}
\begin{thebibliography}{DDGR07}

\bibitem[AI06]{AI06}
A.~Andoni and P.~Indyk.
\newblock Efficient algorithms for substring near neighbor problem.
\newblock In {\em Proc.\ 17th Ann.\ ACM-SIAM Symposium on Discrete Algorithm},
  pages 1203--1212, 2006.

\bibitem[AI08]{AI08}
A.~Andoni and P.~Indyk.
\newblock {Near-optimal hashing algorithms for approximate nearest neighbor in
  high dimensions}.
\newblock {\em Communications of the ACM}, 51(1):117--122, 2008.

\bibitem[Buh01]{Buh01}
J.~Buhler.
\newblock Efficient large-scale sequence comparison by locality-sensitive
  hashing.
\newblock {\em Bioinformatics}, 17(5):419--428, 2001.

\bibitem[CDF{\etalchar{+}}01]{CDFGIMUY01}
E.~Cohen, M.~Datar, S.~Fujiwara, A.~Gionis, P.~Indyk, R.~Motwani, J.D. Ullman,
  and C.~Yang.
\newblock Finding interesting associations without support pruning.
\newblock {\em IEEE Transactions on Knowledge and Data Engineering},
  13(1):64--78, 2001.

\bibitem[DDGR07]{DDGR07}
A.~Das, M.~Datar, A.~Garg, and S.~Rajaram.
\newblock Google news personalization: scalable online collaborative filtering.
\newblock In {\em Proc.\ 16th Intl.\ Conf.\ on World Wide Web}, pages 271--280,
  2007.

\bibitem[DIIM04]{DIIM04}
M.~Datar, N.~Immorlica, P.~Indyk, and V.~S. Mirrokni.
\newblock Locality-sensitive hashing scheme based on p-stable distributions.
\newblock In {\em Proc.\ 20th Ann.\ Symposium on Computational Geometry}, pages
  253--262, New York, NY, USA, 2004.

\bibitem[GIM99]{GIM99}
A.~Gionis, P.~Indyk, and R.~Motwani.
\newblock {Similarity search in high dimensions via hashing}.
\newblock In {\em Proc.\ 25th Intl.\ Conf.\ on Very Large Data Bases}, 1999.

\bibitem[IM98]{IM98}
P.~Indyk and R.~Motwani.
\newblock {Approximate nearest neighbors: towards removing the curse of
  dimensionality}.
\newblock In {\em Proc.\ 13th Ann.\ ACM Symposium on Theory of Computing},
  pages 604--613, 1998.

\bibitem[Ind01]{Ind01-thesis}
P.~Indyk.
\newblock {\em {High-dimensional computational geometry}}.
\newblock PhD thesis, Stanford University, 2001.

\bibitem[Ind04]{Ind04}
P.~Indyk.
\newblock {Nearest neighbors in high-dimensional spaces}.
\newblock {\em Handbook of Discrete and Computational Geometry}, pages
  877--892, 2004.

\bibitem[Ind09]{Ind09}
P.~Indyk.
\newblock Personal communication, 2009.

\bibitem[MNP07]{MNP07}
R.~Motwani, A.~Naor, and R.~Panigrahi.
\newblock Lower bounds on locality sensitive hashing.
\newblock {\em SIAM Journal on Discrete Mathematics}, 21(4):930--935, 2007.

\bibitem[Ney10]{Ney10}
T.~Neylon.
\newblock {A locality-sensitive hash for real vectors}.
\newblock {\em To appear in the 21st Ann.\ ACM-SIAM Symposium on Discrete
  Algorithms}, 2010.

\bibitem[O'D03]{O'D03}
R.~O'Donnell.
\newblock {\em {Computational applications of noise sensitivity}}.
\newblock PhD thesis, Massachusetts Institute of Technology, 2003.

\bibitem[Pan06]{Pan06}
R.~Panigrahy.
\newblock {Entropy based nearest neighbor search in high dimensions}.
\newblock In {\em Proc.\ 17th Ann.\ ACM-SIAM Symposium on Discrete Algorithm},
  page 1195, 2006.

\bibitem[PTW08]{PTW08}
R.~Panigrahy, K.~Talwar, and U.~Wieder.
\newblock {A geometric approach to lower bounds for approximate near-neighbor
  search and partial match}.
\newblock In {\em Proc.\ 49th Ann.\ IEEE Symposium on Foundations of Computer
  Science}, pages 414--423. IEEE Computer Society, 2008.

\bibitem[RPH05]{RPH05}
D.~Ravichandran, P.~Pantel, and E.~Hovy.
\newblock Randomized algorithms and {NLP}: Using locality sensitive hash
  functions for high speed noun clustering.
\newblock In {\em Proc.\ 43rd Ann.\ Meeting of the Association for
  Computational Linguistics}, pages 622--629, Ann Arbor, Michigan, June 2005.
  Association for Computational Linguistics.

\bibitem[SVD03]{SVD03}
G.~Shakhnarovich, P.~Viola, and T.~Darrell.
\newblock {Fast pose estimation with parameter-sensitive hashing}.
\newblock In {\em Proc.\ 9th Ann.\ IEEE Intl.\ Conf.\ on Computer Vision},
  pages 750--757. Citeseer, 2003.

\bibitem[TT07]{TT07}
K.~Terasawa and Y.~Tanaka.
\newblock {Spherical {LSH} for approximate nearest neighbor search on unit
  hypersphere}.
\newblock {\em Lecture Notes in Computer Science}, 4619:27, 2007.

\end{thebibliography}

\end{document}